\documentclass{article}
\usepackage{times}

\usepackage[utf8]{inputenc}
\usepackage{amsmath}
\usepackage{amsfonts}
\usepackage{amssymb}
\usepackage{tikz,xspace}
\usepackage{amsthm}
\usepackage[round]{natbib} 
\title{Structure in Dichotomous Preferences\thanks{A preliminary version appeared in the proceedings of IJCAI 2015, the International Joint Conference on Artificial Intelligence \citep{aaai/ElkindL15-dichpref}.}}
\author{Edith Elkind\\University of Oxford, UK \\ elkind@cs.ox.ac.uk 
\and Martin Lackner\\University of Oxford, UK \\martin.lackner@cs.ox.ac.uk}
\date{}

\newcommand{\calP}{{\cal P}}

\newcommand{\calO}{{\cal O}}
\newcommand{\calN}{{\cal N}}
\newcommand{\calW}{{\cal W}}

\newcommand{\poly}{{\mathrm{poly}}}
\newcommand{\vecw}{{\mathbf{w}}}
\newcommand{\vecr}{{\mathbf{r}}}
\newcommand{\pav}{\ensuremath{\mathit{PAV}}\xspace}
\newcommand{\mav}{\ensuremath{\mathit{MAV}}\xspace}

\theoremstyle{plain}
\newtheorem{theorem}{Theorem}
\newtheorem{proposition}[theorem]{Proposition}
\newtheorem{lemma}[theorem]{Lemma}

\newtheorem{conjecture}[theorem]{Conjecture}
\newtheorem*{theorem*}{Theorem}
\newtheorem*{fact*}{Fact}
\newtheorem*{proposition*}{Proposition}
\theoremstyle{definition}
\newtheorem{definition}{Definition}
\newtheorem*{definition*}{Definition}
\newtheorem{example}{Example}
\theoremstyle{remark}

\newtheorem*{remark*}{Remark}

\begin{document}

\maketitle

\begin{abstract}\noindent
Many hard computational social choice problems are known to become tractable
when voters' preferences belong to a restricted domain, such
as those of single-peaked or single-crossing preferences. 
However, to date, all algorithmic results of this type
have been obtained for the setting where each voter's preference list
is a total order of candidates. The goal of this paper
is to extend this line of research to the setting where voters'
preferences are {\em dichotomous}, i.e., each voter approves
a subset of candidates and disapproves the remaining candidates.
We propose several analogues of the notions of single-peaked and single-crossing
preferences for dichotomous profiles and investigate the relationships among them.
We then demonstrate that for some of these notions the respective
restricted domains admit efficient algorithms for computationally hard 
approval-based multi-winner rules. 
\end{abstract}

\section{Introduction}
Preference aggregation is a fundamental problem in social choice, which
has recently received a considerable amount of attention from the AI community.
In particular, an important research question in {\em computational social choice}
\citep{comsoc-book} 
is the complexity of computing the output of various 
preference aggregation procedures. While for most common single-winner 
rules winner determination is easy,
many attractive rules that output a committee (a fixed-size set of winners)
or a ranking of the candidates are known to be computationally hard.

There are several ways to circumvent these hardness results, such as using approximate
and parameterized algorithms. These standard algorithmic approaches are complemented  
by an active stream of research that analyzes the computational complexity of voting rules
on {\em restricted preference domains}, such as the classic domains of 
single-peaked \citep{bla:b:sp}
or single-crossing \citep{mir:j:single-crossing} preferences. This research direction was 
popularized by \citet{wal:c:uncertainty}
and \citet{fal-hem-hem-rot:j:sp}, 
and has lead to a number of efficient algorithms
for winner determination under prominent voting rules as well as for manipulation and control,
which can be used when voters' preferences belong to one of these restricted domains
\citep{wal:c:uncertainty,fal-hem-hem-rot:j:sp,BBHH10,FHH14,BSU13a,sko-yu-fal:j:mwsc,MF14}.

To the best of our knowledge, this line of work only considers settings
where voters' preferences are given by total orders over the set of candidates;
indeed, this is perhaps the most widely studied setting in the area
of computational social choice. However, computationally complex
preference aggregation problems may also arise when voters' preferences are
{\em dichotomous}, i.e., each voter approves a subset of the candidates 
and disapproves the remaining candidates. Committee selection rules
for voters with dichotomous preferences, or {\em approval-based}
rules, have recently attracted some attention
from the computational social choice community, and for two prominent such rules
(specifically, Proportional Approval Voting (PAV) \citep{MaKi12a}
and Maximin Approval Voting (MAV) \citep{BKS07a})
computing the winning committee is known to be NP-hard \citep{AGG+14a,LMM07a}. It is therefore
natural to ask if one could identify a suitable analogue of 
single-peaked/single-crossing preferences for the the dichotomous setting,
and design efficient algorithms for approval-based rules 
over such restricted dichotomous preference domains.

To address this challenge, in this paper we propose and explore a number
of domain restrictions for dichotomous preferences that build
on the same intuition as the concepts of single-peakedness and 
single-crossingness. Some of our restricted domains  
are defined by embedding voters or candidates into the real line,
and requiring that the voters' preferences over the candidates 
``respect'' this embedding; others are obtained by viewing 
dichotomous preferences as weak orders and requiring them to admit 
a refinement that has a desirable structural property.
Surprisingly, these approaches lead
to a large number of concepts that are pairwise non-equivalent
and capture different aspects of our intuition about what it means
for preferences to be ``one-dimensional''.
We analyze the relationships among these restricted preference domains, 
(see Figure~\ref{fig:relations} for a summary), and discuss the complexity
of detecting whether a given dichotomous profile belongs
to one of these domains. We then demonstrate that considering
these domains is useful from the perspective of algorithm design, 
by providing polynomial-time and FPT algorithms for PAV and MAV 
under some of these domain restrictions. 

\section{Basic Definitions}
Let $C=\{c_1,\dots,c_m\}$ be a finite set of candidates.
A {\em partial order} $\succ$ over $C$ is a reflexive, antisymmetric
and transitive binary relation on $C$; a partial order $\succ$
is said to be {\em total} if for each $c,d\in C$ we have $c\succ d$ or $d\succ c$. 
We say that a partial order $\succ$ over $C$ is a {\em dichotomous weak order}
if $C$ can be partitioned into two disjoint sets $C^+$ and $C^-$
(one of which may be empty) so that $c\succ d$ for each $c\in C^+, d\in C^-$
and the candidates within $C^+$ and $C^-$ are incomparable under $\succ$.

An {\em approval vote} on $C$ is an arbitrary subset of $C$.
We say that an approval vote $v$ is {\em trivial} if $v=\emptyset$
or $v=C$. A {\em dichotomous profile} $\calP=(v_1,\dots,v_n)$ is a list of $n$
approval votes; we will refer to $v_i$ as the vote of voter $i$.
We write $\overline{v_i}=C\setminus v_i$.
We associate an approval vote $v_i$ with 
the dichotomous weak order $\succ_{v_i}$ that satisfies $c\succ_{v_i} d$ 
if and only if $c\in v_i$, $d\in \overline{v_i}$. Note that $v_i=\emptyset$
and $v_i=C$ correspond to the same dichotomous weak order, namely
the empty one. 

A partial order $\succ'$ over $C$ is a {\em refinement}
of a partial order $\succ$ over $C$ if for every $c,d\in C$
it holds that $c\succ d$ implies $c\succ' d$. A 
profile $\calP'=(\succ_1,\dots,\succ_n)$ of total orders
is a {\em refinement} of a dichotomous profile 
$\calP=(v_1,\dots, v_n)$ if $\succ_i$ is a refinement
of $\succ_{v_i}$ for each $i=1,\dots,n$. 

Let $\lhd$ be a total order over $C$.
A total order $\succ$ over $C$ is said to be {\em single-peaked
with respect to $\lhd$} if for any triple of candidates 
$a,b,c\in C$ with $a\lhd b\lhd c$ or $c\lhd b\lhd a$ 
it holds that $a\succ b$ implies $b\succ c$. A profile 
$\calP$ of total orders over $C$ is said to be {\em single-peaked}
if there exists a total order $\lhd$ over $C$ such that all orders in $\calP$
are single-peaked with respect to $\lhd$. 

A profile $\calP=(\succ_1,\dots,\succ_n)$ of total orders over $C$
is said to be {\em single-crossing with respect to the given 
order of votes} if for every pair of candidates $a,b\in C$
such that $a\succ_1 b$ all votes where $a$ is preferred to $b$
precede all votes where $b$ is preferred to $a$;
$\calP$ is {\em single-crossing} if the votes in $\calP$
can be permuted so that it becomes single-crossing with respect
to the resulting order of votes.

A profile $\calP=(\succ_1,\dots,\succ_n)$ of total orders over $C$
is said to be {\em $1$-Euclidean} if there is a mapping $\rho$ of
voters and candidates into the real line such that $c\succ_i d$
if and only if $|\rho(i)-\rho(c)|<|\rho(i)-\rho(d)|$.
A $1$-Euclidean profile is both single-peaked and single-crossing.

\section{Preference Restrictions}\label{sec:constraints}

We will now define a number of constraints that a dichotomous profile may satisfy.
Most of these constraints can be divided into two basic groups: 
those that are based on ordering voters and/or candidates on the line
and requiring the votes to respect this order (this includes
VEI, VI, CEI, CI, DE, and DUE), and those that are based
on viewing votes as weak orders and asking if there is a 
single-peaked/single-crossing/1-Euclidean profile of total orders
that refines the given profile 
(this includes PSP, PSC, and PE); we remark that the study of the latter type
of constraints was initiated by \citet{La14}. 
We will also consider constraints that are based on partitioning voters/candidates
(2PART and PART), as well as two constraints (WSC and SSC) that have been introduced
in a recent paper of \citet{EFLO15} in order to 
understand the best way of extending the single-crossing property 
to weak orders.

Fix a profile $\calP=(v_1,\dots,v_n)$ over $C$.

\begin{enumerate}
\item
\emph{$2$-partition (2PART)}:
We say that $\calP$ satisfies 2PART
if $\calP$ contains only two distinct votes $v,v'$,
and $v\cap v'=\emptyset$, $v\cup v'=C$.
\item
\emph{Partition (PART)}:
We say that $\calP$ satisfies PART 
if $C$ can be partitioned into pairwise disjoint subsets $C_1,\dots,C_\ell$
such that $\{v_1,\dots,v_n\}=\{C_1,\dots,C_\ell\}$ (i.e., each voter in $\calP$
approves one of the sets $C_1,\dots,C_\ell$).
Note that this constraint contains as a special case profiles where every voter approves of exactly one candidate.
\item 
\emph{Voter Extremal Interval (VEI)}: 
We say that $\calP$ satisfies VEI if the voters in $\calP$
can be reordered so that for every candidate $c$ 
the voters that approve $c$ form a prefix or a suffix of the ordering.
Equivalently, both the voters who approve $c$ and the voters
who disapprove $c$ form an interval of that ordering.
See Figure~\ref{fig:vei} for an example.
\begin{figure}
  \centering
  \hfill
  \begin{minipage}[b]{0.45\textwidth}
	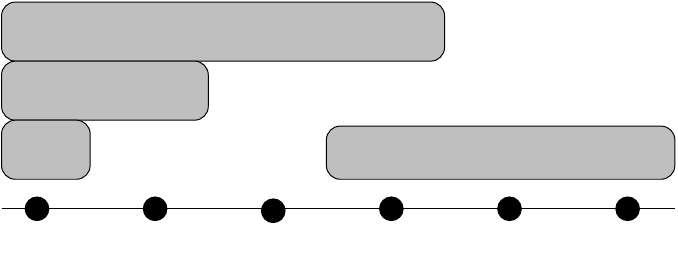
	\caption{Voter Extremal Interval}
	\label{fig:vei}
  \end{minipage}
  \begin{minipage}[b]{0.02\textwidth}\ 
  \end{minipage}
  \begin{minipage}[b]{0.45\textwidth}
	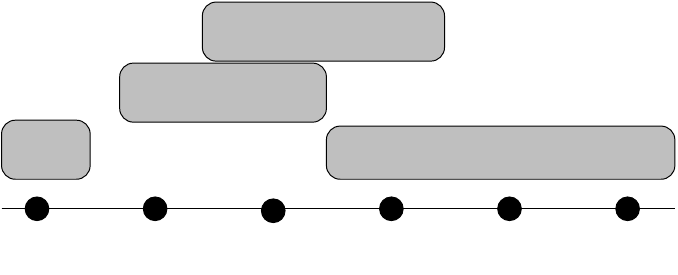
	\caption{Voter Interval}
	\label{fig:vi}
  \end{minipage}
\end{figure}

\item 
\emph{Voter Interval (VI)}: 
We say that $\calP$ satisfies VI if the voters in $\calP$
can be reordered so that for every candidate $c$
the voters that approve $c$ form an interval of that ordering.
See Figure~\ref{fig:vi} for an example.

\emph{Candidate Extremal Interval (CEI)}: 
We say that $\calP$ satisfies CEI if candidates in $C$
can be ordered so that each of the sets $v_i$
forms a prefix or a suffix of that ordering.
Equivalently, both $v_i$ and $\overline{v_i}$
form an interval of that ordering.
See Figure~\ref{fig:cei} for an example.
\begin{figure}[b]
  \centering
  \hfill
  \begin{minipage}[b]{0.45\textwidth}
	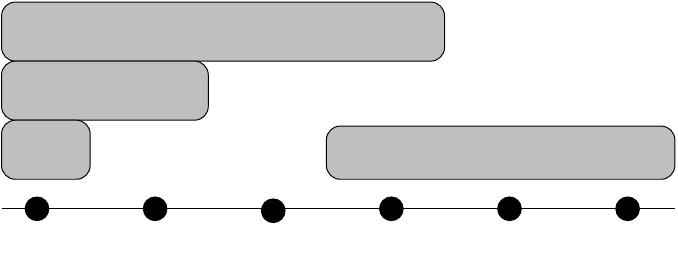
	\caption{Candidate Extremal Interval}
	\label{fig:cei}
  \end{minipage}
  \begin{minipage}[b]{0.02\textwidth}\ 
  \end{minipage}
  \begin{minipage}[b]{0.45\textwidth}
	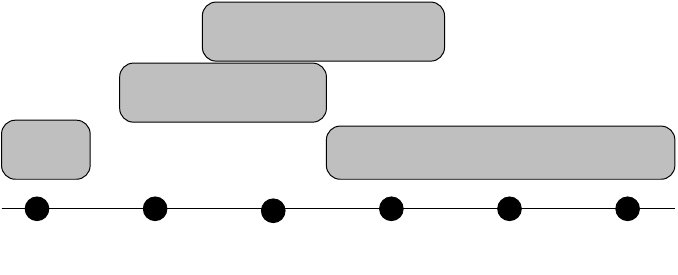
	\caption{Candidate Interval}
	\label{fig:ci}
  \end{minipage}
\end{figure}
\item 
\emph{Candidate Interval (CI)}: 
We say that $\calP$ satisfies CI if candidates in $C$
can be ordered so that each of the sets $v_i$
forms an interval of that ordering.
See Figure~\ref{fig:ci} for an example.
\item
\emph{Dichotomous Uniformly Euclidean (DUE)}: 
We say that $\calP$ satisfies DUE if there is a mapping
$\rho$ of voters and candidates into the real line and a radius $r$
such that for every voter $i$ it holds that
$v_i=\{c: |\rho(i)-\rho(c)|\le r\}$.
\item
\emph{Dichotomous Euclidean (DE)}:
We say that $\calP$ satisfies DE if there is a mapping
$\rho$ of voters and candidates into the real line
such that for every voter $i$ there exists a radius $r_i$
with $v_i=\{c: |\rho(i)-\rho(c)|\le r_i\}$.
\item 
\emph{Possibly single-peaked (PSP)}: 
We say that $\calP$ satisfies PSP if 
there is a single-peaked profile of total orders $\calP'$ that is a refinement of $\calP$.
\item  
\emph{Possibly single-crossing (PSC)}:
We say that $\calP$ satisfies PSC if 
there is a single-crossing profile of total orders $\calP'$ that is a refinement of $\calP$.
\item 
\emph{Possibly Euclidean (PE)}:
We say that $\calP$ satisfies PE if
there is a $1$-Euclidean profile of total orders $\calP'$ that is a refinement of $\calP$. 
\item
\emph{Seemingly single-crossing (SSC)}:
We say that $\calP$ satisfies SSC
if the voters in $\calP$ can be reordered so that
for each pair of candidates $a,b\in C$ it holds that
either all votes $v_i$ with $a\in v_i$, $b\not\in v_i$
precede all votes $v_j$ with $a\not\in v_j$, $b\in v_j$
or vice versa.
\item
\emph{Weakly single-crossing (WSC)}:
We say that $\calP$ satisfies WSC if
the voters in $\calP$ can be reordered so that
for each pair of candidates $a,b\in C$ it holds that
each of the vote sets $V_1=\{v_i: a\in v_i, b\not\in v_i\}$,
$V_2=\{v_i: a\not\in v_i, b\in v_i\}$, $V_3=\{v\in\calP: v\not\in V_1\cup V_2\}$
forms an interval of this ordering, with $V_3$ appearing between $V_1$ and $V_2$. 
\end{enumerate}


\subsection{Relations}
The relationships among the properties defined above are depicted 
in Figure~\ref{fig:relations}, where arrows indicate containment, i.e., more restrictive notions 
are at the top. All these containments are strict.

\begin{figure}
\begin{center}
\begin{tikzpicture}[scale=2]

\def \n {5}
\def \radius {3cm}
\def \margin {8} 

\node (top) [draw] at (2,4) {2PART};
\node (part) [draw] at (3,3.5) {PART};
\node (z) [draw] at (2,2) {PSC=SSC};
\node (a) [draw] at (0.75,3.5) {VEI};
\node (c) [draw] at (2.25,3.5) {CEI};
\node (b) [draw] at (1.5,3.5) {WSC};
\node (e) [draw] at (2,3) {DUE};
\node (f) [draw] at (2.75,2.5) {CI=DE=PSP=PE};
\node (h) [draw] at (1.25,2.5) {VI};

\draw[->] (top) -- (a);
\draw[->] (top) -- (c);
\draw[->] (top) -- (b);
\draw[->] (top) -- (part);
\draw[->] (part) -- (e);
\draw[->] (f) -- (z);
\draw[->,dashed] (b) -- (a) ;
\draw[->,dashed] (b) -- (c);
\draw[->] (a) -- (e);
\draw[->] (b) -- (e);
\draw[->] (c) -- (e);
\draw[->] (e) -- (f);
\draw[->] (e) -- (h);
\draw[->] (h) -- (z);
\end{tikzpicture}
\end{center}
\caption{Relations between notions of structure. Dashed lines indicate that the respective containment 
holds only subject to additional conditions.}
\label{fig:relations}
\end{figure}
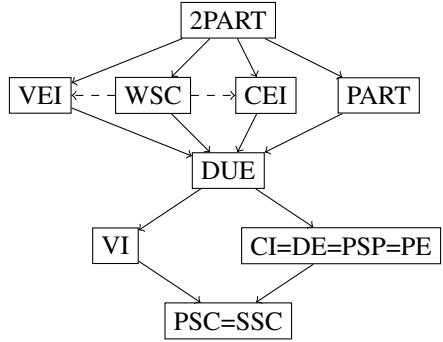

The four arrows at the top level of the diagram are immediate:
any profile with at most two distinct votes 
where each candidate is approved in at least one of these votes
satisfies VEI, CEI and WSC, and by definition 2PART is a special case of PART.

To understand the arrows in the next level,
we first characterize the dichotomous profiles that are weakly single-crossing.

\begin{lemma}\label{lem:WSC-char}
A dichotomous profile $\calP$ satisfies WSC if and only if 
there exist three votes $u, v, w$ such that
\begin{itemize}
\item[(1)] 
for every $v_i\in\calP$ it holds that $\succ_{v_i}\in\{\succ_u,\succ_v,\succ_w\}$, and 
\item [(2)]
$\succ_v$ is equal to either $\succ_{u\cap w}$ or $\succ_{u\cup w}$.
\end{itemize}
\end{lemma}
\begin{proof}[Proof sketch]
It is easy to check that every profile satisfying (1)--(2) satisfies WSC. 
For the converse direction,
assume without loss of generality that
the ordering of the votes $v_1\sqsubset v_2\sqsubset \dots\sqsubset v_n$ 
witnesses that $\calP$ satisfies WSC.
Let $u=v_1$, $w=v_n$, and  
set $C_1=u\cap w$, $C_2=u\cap \overline{w}$, $C_3=\overline{u}\cap w$,
$C_4= \overline{u}\cap \overline{w}$. The WSC property implies that 
for every $\ell=1,2,3,4$, every $a,b\in C_\ell$, and every $v_i\in\calP$ we have $a\in v_i$ if and only
if $b\in v_i$, i.e., candidates in each $C_\ell$ occur as a block in all votes.
Note that $v_1=u=C_1\cup C_2$, $v_n=w=C_1\cup C_3$.

Suppose that $C_1,C_4\neq\emptyset$. Then $C_1\subseteq v_i$, $C_4\subseteq \overline{v_i}$
for all $v_i\in\calP$. 
Indeed, fix a pair of candidates $a\in C_1$, $b\in C_4$. Both the first and the last voter 
strictly prefer $a$ to $b$, and therefore so do all other voters.
Thus, if $\calP$ contains a vote $v_i\neq u,w$, it has to be the case that
$v_i=C_1=u\cap w$ or $v_i=C_1\cup C_2\cup C_3=u\cup w$; moreover, if both of these
votes occur simultaneously and are distinct from each other and $u,w$
(i.e., $C_2,C_3\neq\emptyset$), the WSC property is violated. 
Indeed, suppose that $v_i=C_1$, $v_j=C_1\cup C_2\cup C_3$.
Fix candidates $a\in C_1$, $b\in C_4$. 
If $v_i$ appears before $v_j$, consider a candidate $c\in C_2$:
we get a contradiction as voters $v_1$ and $v_j$
are indifferent between $a$ and $c$, but $v_i$ strictly prefers $a$ to $c$.
If $v_i$ appears after $v_j$, consider a candidate $d\in C_3$:
we get a contradiction as voters $v_1$ and $v_i$
are indifferent between $d$ and $b$, but $v_j$ strictly prefers $d$ to $b$.
When 
$C_1$ or $C_4$ is empty, the analysis is similar; note, however, that trivial votes 
($v_i=C$ and $v_i=\emptyset$) may alternate arbitrarily without violating the WSC property
(this is why the lemma is stated in terms of weak orders rather than approval votes).
\end{proof}

We can now show that under mild additional conditions 
(no trivial voters/candidates) WSC implies VEI and CEI. 

\begin{proposition}\label{prop:VEI-char}
Let $\calP$ be a dichotomous profile that either contains only two distinct votes or contains no 
vote $v_i$ with $v_i=\emptyset$. If $\calP$ satisfies WSC, then it satisfies VEI.
\end{proposition}
\begin{proof}
Assume without loss of generality that $\calP$ 
satisfies WSC with respect to an ordering of voters $v_1\sqsubset\dots\sqsubset v_n$,
and let $u=v_1$, $w=v_n$.
We will show that $\calP$ satisfies VEI with respect to $\sqsubset$. If $\calP$
only contains two distinct votes, this claim is immediate, so assume
that $\emptyset\not\in\calP$. Consider a vote $v\in\calP$ that is distinct from $u$ and $w$.
Since $\emptyset\not\in\calP$, by Lemma~\ref{lem:WSC-char} there exist $i,j$ with $1<i<j<n$ such that 
$v_k=u$ for $k<i$, $v_k=v$ for $k=i,\dots,j$, $v_k=w$ for $k>j$, and $v\in\{u\cup w,u\cap w\}$. 
Suppose first that $v=u\cap w$. Then candidates in $u\cap w$ are approved by all voters,
candidates in $u\setminus w$ are approved by the first $i-1$ voters,
candidates in $w\setminus u$ are approved by the last $n-j$ voters,
and the remaining candidates are not approved by anyone.
On the other hand, if $v=u\cup w$, then candidates in $u\cap w$ are approved by all voters,
candidates in $u\setminus w$ are approved by the first $j$ voters,
candidates in $w\setminus u$ are approved by the last $n-i+1$ voters,
and the remaining candidates are not approved by anyone.  
\end{proof}

\noindent The condition that the profile must not contain $\emptyset$ is 
necessary: the profile $(\{a,b\}, \emptyset, \{b,c\})$ 
satisfies WSC, but not VEI.

\begin{proposition}\label{prop:CEI-char}
Let $\calP$ be a dichotomous profile that either contains only two distinct votes or in which every 
candidate is approved in at least one vote and disapproved in at least one vote. If $\calP$ 
satisfies WSC, then it satisfies CEI.
\end{proposition}
\begin{proof}
Suppose that $\calP$ is WSC with respect to some ordering of voters;
let $u$ and $w$ be, respectively, the first and the last vote in this ordering.
If $\calP$ contains a trivial vote, it contains at most two non-trivial votes,
in which case the claim is obvious. Thus, assume that it contains no trivial votes. 
Then we have $u\cap w=\emptyset$ (any candidate in $u\cap w$ would be approved by all voters)
and $\overline{u}\cap\overline{w}=\emptyset$ (any candidate in $\overline{u}\cap\overline{w}$
 would be disapproved by all voters). It is now easy to see that ordering the candidates
so that all candidates approved by $u$ precede all candidates approved by $w$
witnesses that $\calP$ is CEI.
\end{proof}

\noindent To see that conditions of Proposition~\ref{prop:CEI-char} are necessary,
consider the profile $(\{a,b\}, \{b,c\})$ over $\{a,b,c,d\}$
and the profile $(\{a,b\}, \{b\}, \{b,c\})$ over $\{a,b, c\}$:
both of these profiles satisfy WSC, but not CEI. 

Interestingly, requiring a dichotomous profile to satisy WSC, CEI and VEI simultaneously,
turns out to be very demanding: we obtain 2-partition profiles. 

\begin{proposition}
A dichotomous profile is WSC, CEI and VEI if and only if it is a $2$-partition.
\end{proposition}
\begin{proof}
It is immediate that a $2$-partition profile is WSC, CEI, and VEI.
For the converse direction, 
let $\calP$ be a CEI, VEI and WSC profile. By Lemma~\ref{lem:WSC-char}, $\calP$ 
contains at most three distinct votes $u,v,w$ with $v=u\cap w$ or 
$v=u\cup w$. Since $\calP$ is CEI, we know from Lemma~\ref{prop:CEI-char} that every 
candidate is approved at least once. Hence $u\cup w = C$. Furthermore, 
every candidate is disapproved at least once. 
Thus, $u\cap w=\emptyset$, since this intersection is also approved by $v$. 
Thus, $v$ is a trivial vote. This is possible because of 
Lemma~\ref{prop:VEI-char} and hence $v$ does not appear in $\calP$. We have 
shown that $\calP$ is a $2$-partition profile.
\end{proof}

Next, we will relate CEI and VEI to DUE.

\begin{proposition}\label{prop:CEI/VEI->DUE}
If a dichotomous profile $\calP$ 
satisfies CEI or VEI, then it satisfies DUE.
\end{proposition}
\begin{proof}
Suppose first that $\calP$ satisfies CEI with respect to the ordering $c_1\lhd\dots\lhd c_m$
of candidates. Map the candidates into the real line by setting $\rho(c_i)=i$, and let $r=m$.
We can now place each voter $i$ to the left or to the right of all candidates at an appropriate
distance so that the set of candidates within distance $r$ from him coincides with $v_i$.
For VEI the argument is similar: if $\calP$ satisfies VEI with respect to the ordering 
$v_1\sqsubset\dots\sqsubset v_n$ of voters, we place voters on the real line according
to $\rho(i)=i$, let $r=n$, and place each candidate to the left or to the right of all voters
at an appropriate distance.
\end{proof}

The proof that WSC implies DUE is also based on our characterization of WSC preferences.

\begin{proposition}
If a dichotomous profile $\calP$
satisfies WSC, then it satisfies DUE.
\end{proposition}
\begin{proof}
Clearly empty votes can be ignored when checking whether a profile satisfies DUE,
so assume $\calP$ contains to empty votes.
Then it contains at most three distinct votes $u$, $v$, $w$ with 
$v=u\cap w$ or $v=u\cup w$. Set $\rho(c)=1$ for $c\in u\setminus w$,
$\rho(c)=2$ for $c\in u\cap w$, $\rho(c)=3$ for $c\in w\setminus u$,
$\rho(c)=10$ for $c\not\in u\cup w$.
We set $r=1$ if $v=u\cap w$ and $r=2$ if $v=u\cup w$, and position the voters
accordingly.
\end{proof}

The last arrow on this level is from PART to DUE: here, the containment is straightforward,
as the candidates approved by each voter can be placed as a block on the axis, with the respective
voter(s) placed in the center of this block.

\begin{proposition}\label{prop:DUE->CI/VI}
If a dichotomous profile $\calP$ satisfies
DUE then it satisfies both VI and CI.
The converse direction does not hold: there are profiles that satisfy VI and CI but not DUE.
\end{proposition}
\begin{proof}
Since $\calP$ satisfies DUE, we have an embedding $\rho$ of votes and candidates into the real line.
For VI, we order voters as induced by the $\rho$ mapping; the voters approving some candidate form an interval on this induced order.
For CI, we order candidates as induced by the $\rho$ mapping; voters always approve a single interval on this ordering.

For showing that the converse direction does not hold, consider the profile $(\{a,b,c\}$, $\{b,c,d\}$, $\{b\}$, $\{c\})$.
Towards a contradiction assume that $\rho$ is a mapping of voters and candidates into the real line that witnesses the DUE property for a fixed radius $r$.
The given profile satisfies CI only with respect to the orders $a \lhd b \lhd c\lhd d$, $a \lhd c \lhd b\lhd d$ and their reverses.
Since the profile is symmetric with respect to $a$ and $d$ and with respect to $b$ and $c$, we can assume without loss of generality that $\rho$ orders candidates as the order $a \lhd b \lhd c\lhd d$ does.
Then it has to hold that $|\rho(a)-\rho(c)|\leq r$ since $a$ and $c$ appear in the same vote.
However, due to the vote $\{b\}$, it also has to hold that $|\rho(a)-\rho(c)|> r$; this is a contradiction.
\end{proof}

We see that similar to total orders, where the intersection of the single-peaked and 
the single-crossing domain is a strict subset of the 1-Euclidean domain (see discussion in \citep{DoignonF94,elk-fal-sko:c:spsc}), for dichotomous preferences also VI intersected with CI does not yield DUE.
The next results shows that the classes of CI, DE, PSP and PE preferences coincide.

\begin{proposition}\label{prop:PE-PSP}
Let $\calP$ be a dichotomous profile. Then the following conditions are equivalent:
(a) $\calP$ satisfies PE (b) $\calP$ satisfies PSP 
(c) $\calP$ satisfies CI (d) $\calP$ satisfies DE.
\end{proposition}
\begin{proof}[Proof sketch]
Suppose $\calP$ satisfies PE, and let $\calP'$ be a refinement of $\calP$ that, together
with a mapping $\rho$, witnesses this. Then $\calP'$ is single-peaked and therefore
$\calP$ satisfies PSP. If $\calP$ satisfies PSP, as witnessed by a refinement $\calP'$ and an axis
$\lhd$, then $\calP$ satisfies CI with respect to $\lhd$. If $\calP$ satisfies CI with respect
to an order $\lhd$ of candidates, we can map the candidates into the real axis
in the order suggested by $\lhd$ so that the distance between every 
two adjacent candidates is $1$. We can then choose an appropriate approval radius  
and position for each voter. Finally, if $\calP$ satisfies DE, as witnessed by a mapping $\rho$,
we can use this mapping to construct a refinement of $\calP$; by construction,
this refinement is 1-Euclidean (we may have to modify $\rho$ slightly to avoid ties). 
\end{proof}

Also, every PE profile is PSC since every 1-Euclidean refinement is also single-crossing. 
Interestingly, the converse is not true.
\begin{example}
Consider the profile $\calP = (\{a,b\},\{a,c\},\{b,c\})$ over $C=\{a,b,c\}$.
It satisfies PSC, as witnessed by the single-crossing refinement
$(a\succ b\succ c, c\succ a\succ b, c\succ b\succ a)$. However,
in every refinement of $\calP$ the first voter ranks $c$ last, the second voter ranks $b$ last,
and the third voter ranks $a$ last. Thus, no such refinement can be single-peaked,
and, consequently, no such refinement can be 1-Euclidean.
\end{example}

The equivalence between PSC and SSC is not entirely obvious: while it is clear
that a profile that violates SSC also violates PSC, to prove the converse
one needs to use an argument similar to the proof of Theorem~4 in \citep{EFLO15}.
This has been shown in the extended version of \citep{EFLO15}.

\begin{proposition}
If a dichotomous profile $\calP$ satisfies VI, it also satisfies SSC.  
\end{proposition}
\begin{proof}
Assume that an VI profile is not SSC. Since it is not SSC, for every ordering of votes $\sqsubset$ there are two
candidates $a\succ b$ and votes $v_i\sqsubset v_j\sqsubset v_k$ such that $v_i:a\succ b$, $v_j:b\succ a$ and
$v_k:a\succ b$. This implies, however, that for every $\sqsubset$ there is a candidate $a$ and votes $v_i\sqsubset
v_j\sqsubset v_k$ such that $v_i$ and $v_k$ approve of $a$ and $v_j$ disapproves $v_j$. This contradicts our
assumption that the given profile is VI. 
\end{proof}

We are now going to list the remaining counter-examples for containment and thus show that the arrows in Figure~\ref{fig:relations} indeed indicate strict containment.
\begin{itemize}
\item CI$\not\rightarrow$VI: Consider $(\{a,b,c\},\{a\},\{b\},\{c\})$. This profile is CI with respect to $a\lhd b\lhd c$. It is not VI since the vote $\{abc\}$ would have to be placed next to $\{a\},\{b\},\{c\}$.
\item VI$\not\rightarrow$CI: Consider $(\{a,b\},\{a,c\},\{a,d\})$. This profile is VI for the given order of voters. It is not CI since $a$ has to lie next to $b,c,d$.
\item VEI$\not\rightarrow$CEI: Consider $(\{a,b\},\{a,d\},\{c,d\})$. This profile is VEI for the given order of voters. It is not CEI since $a$ has to lie next to $b$ and $d$ and $c$ has to lie next to $d$. So $b \lhd a \lhd d\lhd c$ is the only order witnessing CI, but the vote $\{a,d\}$ is not an extremal interval on this order.
\item CEI$\not\rightarrow$VEI: Consider $(\{a,b\},\{a\},\{c\},\{b,c\})$. All votes are extremal intervals on the order $a\lhd b \lhd c$. The profile is however not VEI since $\{a,b\}$ has to lie next to $\{a\}$ and next to $\{bc\}$ and $\{c\}$ next to $\{bc\}$. So we obtain $\{c\} \sqsubset \{b,c\} \sqsubset \{a,b\} \sqsubset \{a\}$ as the only order witnessing VI, but this order does not satisfy VEI (consider candidate $b$).
\item PART$\not\rightarrow$VEI, CEI, WSC: Consider the PART profile $\{a\},\{b\},\{c\}$.
\end{itemize}
All other counterexamples involving WSC immediately follow from Lemma~\ref{lem:WSC-char} and Proposition~\ref{prop:VEI-char} and \ref{prop:CEI-char};
all missing counterexamples involving PART can be obtained by picking intersecting votes.


\subsection{Unique orders}

If voter's preferences are given by total orders, single-crossing profiles have a unique single-crossing order, i.e., only one specific order and its reverse witness the single-crossing property of the profile.
For single-peaked profiles (of total orders) this is not the case.
The question arises whether a similar phenomenon can be observed for dichotomous profiles.
Clearly, this question only makes sense for profiles with distinct votes (for VI, VEI, WSC) and when all candidates are approved by some vote (for CI and CEI).
Also, by unique we always mean that only one specific order and its reverse witness a certain restriction.

For dichotomous profiles satisfying SC, there is no unique order.
The profile $(\{a\},\{a,b\},\{b,c\})$ is SC and all votes that put $\{b,c\}$ at an outermost position witness the SC property.
Also profiles satisfying VI or CI do not have unique orders witnessing these properties; e.g., consider $\{\},\{a\},\{b\}$ and $\{a\},\{b\},\{c\}$, respectively.

For profiles being WSC, VEI or CEI we can show that their corresponding orders are indeed unique.
For profiles satisfying WSC, this follows from Lemma~\ref{lem:WSC-char}; for profiles satisfying either VEI or CEI the uniqueness can be shown as follows.

\begin{lemma}
For profiles containing distinct votes, VEI orders are unique.
\end{lemma}
\begin{proof}
Without loss of generality assume that $1 \sqsubset\dots \sqsubset n$ be a VEI order.
Assume towards a contradiction that $\sqsubset'$ is another VEI order that is neither $\sqsubset$ nor its reverse.
Consequently, there exist three votes $v_i,v_j,v_k$, $i<j<k$ for which $\sqsubset$ and $\sqsubset'$ disagree on their order in the sense that $v_j$ is not in between $v_i$ and $v_k$ with respect to $\sqsubset'$.
Without loss of generality let us assume $j\sqsubset' i \sqsubset' k$.
Let us consider $C_X$ for every $X\subseteq\{i,j,k\}$ being defined as the set of all candidates approved by the votes corresponding to $X$  but not approved by those corresponding to $\{i,j,k\}\setminus X$.
For example, $C_{ik}$ are those candidates approved by $c_i$ and $c_k$ but not by $c_j$.
Since we have a CEI profile and $i\sqsubset j \sqsubset k$, we know that $C_{ik}=C_j=\emptyset$.
Under our assumption that $\sqsubset'$ is also a VEI ordering with $j\sqsubset' i \sqsubset' k$, we know that $C_{jk}=C_i=\emptyset$.
This implies that the candidate approved by $c_i$ are $C_{ijk}\cup C_{ij}$ and those approved by $c_j$ are also $C_{ijk}\cup C_{ij}$.
This contradicts our assumption that all votes are distinct.
\end{proof}

\begin{lemma}
If all candidates are approved by distinct sets of voters, CEI orders are unique.
\end{lemma}
\begin{proof}
First, let us observe that two candidates that are approved by the same voters certainly are indistinguishable; their positions on the CEI axis are interchangeable.
Thus, our condition is necessary for the lemma to hold.
The proof of this statement is similar to the previous proof.
Without loss of generality assume that $c_1 \lhd\dots \lhd c_m$ be a CEI order.
Assume towards a contradiction that $\lhd'$ is another CEI order that is neither $\lhd$ nor its reverse.
Consequently, there exist three votes $v_i,v_j,v_k$, $i<j<k$ for which $\lhd$ and $\lhd'$ disagree on their order in the sense that $c_j$ is not in between $c_i$ and $c_k$ with respect to $\lhd'$.
Without loss of generality let us assume $c_j\lhd' c_i \lhd' c_k$.
Let us consider $V_X$ for every $X\subseteq\{i,j,k\}$ being defined as the set of all votes that approve the candidates in $X$ and disapprove those in $\{i,j,k\}\setminus X$.
Since we have a VEI profile and $c_i\lhd c_j \lhd c_k$, we know that $V_{ik}=V_j=\emptyset$.
Under our assumption that $\lhd'$ is also a CEI ordering with $c_j\lhd' c_i \lhd' c_k$, we know that $V_{jk}=V_i=\emptyset$.
This implies that the votes that approve $c_i$ are $V_{ijk}\cup V_{ij}$ and the votes approving $c_j$ are $V_{ijk}\cup V_{ij}$.
This contradicts our assumption that all candidates are approved by a distinct set of voters.
\end{proof}


\subsection{Detection}
To exploit the constraints defined in Section~\ref{sec:constraints},
we have developed algorithms that can decide whether a given profile
belongs to one of the restricted domains defined by these constraints. 
Our results are summarized in Table~\ref{tab:compl}. 

\begin{table}
\begin{center}
\begin{tabular}{c|c}
constraint & complexity\\
\hline
2PART & poly (trivial) \\
PART & poly (trivial) \\
VEI & poly ({\sc consecutive 1s}) \\
CEI & poly ({\sc consecutive 1s}) \\
WSC & poly \citep{EFLO15} \\
DUE & poly \citep{woeginger-nederlof:personal-commun}  \\
VI & poly ({\sc consecutive 1s}) \\
CI=DE=PSP=PE & poly ({\sc consecutive 1s}) \\
PSC=SSC & open \\
\end{tabular}
\caption{The complexity of detecting structure in dichotomous profiles}
\label{tab:compl}
\end{center}
\end{table}

Clearly, verifying whether a given profile satisfies 2PART or PART is straightforward.
For most of the remaining problems, we can proceed by a reduction to the classic {\sc Consecutive 1s}
problem \citep{boo-lue:j:cons1s}. This problem asks if the columns of a given $0$-$1$ 
matrix can be permuted in such a way that in each row of the
resulting matrix the $1$s are consecutive, i.e., the $1$s form an interval in each row;
it admits a linear-time algorithm \citep{boo-lue:j:cons1s}.

\begin{theorem}
Detecting whether a dichotomous profile satisfies CEI, CI, VI or VEI is possible in $\calO(m\cdot n)$ time.
\end{theorem}
\begin{proof}
Let $C=\{c_1,c_2,\ldots,c_m\}$ and $\mathcal{P}=(v_1,v_2,\dots,v_n)$.
We construct an instance of \textsc{Consecutive 1s} in slightly different ways, depending on the property we
want to detect. In all cases, we obtain a ``yes''-instance if and only if
the given profile has the desired property. 

Let us start with CI. For each vote, we create one row of the matrix:
for each $i\in[n]$ and $j\in[m]$, the $j$-th entry of the $i$-th row is $1$ if 
$c_j\in v_i$ and $0$ otherwise. In this way, we obtain an $m\times n$ matrix. 
Permuting the columns of this matrix so that $1$s form an interval in each row
is equivalent to permuting candidates so that the set of candidates approved by each voter forms an interval.
For CEI, we combine the matrix for CI with its complement, i.e., we add a second
row for each vote $v_i$, so that the $j$-the entry of that row is $0$ of $c_j\in v_i$ and $1$ otherwise.
A column permutation of the resulting $m\times 2n$ matrix such that $1$s form an interval in each row
corresponds to permuting candidates so that for each voter both the set of her approved candidates
and the set of her disapproved candidates form an interval; this is equivalent to the CEI property. 
For VI it suffices to transpose
the matrix constructed for CI, and for VEI this matrix has to be combined with its complement.  
\end{proof}

For WSC, \citet{EFLO15} provide an algorithm that works for any weak orders
(not just dichotomous ones). They leave the complexity of detecting PSC and SSC as an open problem,
and we have not been able to resolve it for dichotomous weak orders. 
The problem of recognizing DUE preferences has recently been shown to be solvable in polynomial time by \citet{woeginger-nederlof:personal-commun} via a connection to bipartite permutation graphs.


\section{Algorithms for Committee Selection}\label{sec:algorithms}
In this section, we consider two classic approval-based committee selection rules---Proportional
Approval Voting (\pav) and Maximin Approval Voting (\mav)---and argue that
we can design efficient algorithms
for these rules when voters' preferences belong to some of the domains in our list
(for some of the richer domains, we may need to place mild additional restrictions on voters' preferences).

We start by providing formal definitions of these rules.
\begin{definition}\label{def:pav}
Every non-increasing infinite sequence of non-negative reals 
$\vecw=(w_1$, $w_2,\dots)$ that satisfies $w_1=1$ defines a {\em committee selection rule $\vecw$-\pav}.
This rule takes a set of candidates $C$, a dichotomous profile $\calP=(v_1,\dots,v_n)$
and a target committee size $k\le |C|$ as its input. For every size-$k$ subset $W$ of $C$,
it computes its $\vecw$-\pav score as $\sum_{v_i\in \calP} u_\vecw(|W\cap v_i|)$, 
where $u_\vecw(p)=\sum_{j=1}^p w_j$, and outputs a size-$k$ subset with the highest
$\vecw$-\pav score, breaking ties arbitrarily.
The $\vecw$-\pav rule with $\vecw=(1,\frac12,\frac13,\dots)$ is usually referred
to simply as the \pav rule, and we write $u(p)=1+\dots+\frac{1}{p}$.  
\end{definition}
\noindent PAV is of particular interest since it is the only known approval-based committee selection rule that satisfies the \emph{Extendend Justified Representation} property \citep{aaai/AzizBCEFW15}, which intuitively states that every large enough homogenous group has to be represented in the committee.
In what follows we assume that the entries
of $\vecw$ are rational and $w_i$ can be computed in time $\poly(i)$.

\begin{definition}
Given a set of candidates $C$, a dichotomous profile $\calP=(v_1,\dots,v_n)$
and a target committee size $k\le |C|$, the \mav-score of a size-$k$ subset $W$ of $C$
is computed as $\max_{v_i\in \calP} (|W\setminus v_i|+|v_i\setminus W|)$.
\mav outputs a size-$k$ subset with the lowest
\mav score, breaking ties arbitrarily.
\end{definition}

The $\vecw$-\pav rule is defined by \citet{MaKi12a},
see also \citep{Kilg10a}.
Intuitively, under this rule each voter is assumed to derive a utility of $1$ from having
exactly one of his approved candidates in the winning set; his marginal utility from having
more of his approved candidates in the winning set is non-increasing.
The goal of the rule is to maximize the sum of players' utilities. In contrast, \mav
\citep{BKS07a}
has an egalitarian objective: for each candidate committee, it computes the dissatisfaction
of the least happy voter, and outputs a committee that minimizes the quantity.

Computing the winning committee 
under \mav and \pav is NP-hard, see, respectively,
\citep{LMM07a} and \citep{SFL15,AGG+14a}. The hardness result for \pav extends to 
$\vecw$-\pav as long as $\vecw$ satisfies $w_1>w_2$; moreover, it holds even if
each voter approves of at most two candidates or if each candidate
is approved by at most three voters.

We will now show that \pav 
admits an algorithm whose running time is polynomial
in the number of voters and the number of candidates 
if the input profile satisfies CI or VI and,
furthermore, each voter approves at most $s$
candidates or each candidate is approved by at most $d$ voters,
where $s$ and $d$ are given constants. More specifically,
we prove that \pav winner determination for CI and VI preferences
is in FPT with respect to parameter $s$ and in XP with respect to parameter $d$.
For simplicity, we state our results for \pav; however, all of them
can be extended to $\vecw$-\pav. 

In what follows, we write $[x:y]$ to denote the set $\{z\in {\mathbb Z}: x\le z\le y\}$.
 
\begin{theorem}\label{thm:vi-pav-s}
Given a dichotomous profile $\calP=(v_1,\dots,v_n)$ over a candidate set
$C=\{c_1,\dots,c_m\}$ and a target committee size $k$,
if $|v_i|\le s$ for all $v_i\in \calP$ and $\calP$ satisfies VI,
then we can find a winning committee
under \pav in time $\calO(2^{2s}\cdot k\cdot n)$.
\end{theorem}
\begin{proof}
Assume that $\calP$ satisfies VI with respect to the order of voters $v_1\sqsubset\dots\sqsubset v_n$. 
For each triple
$(i,A,\ell)$, where $i\in[1:n]$, $A\subseteq v_i$, and $\ell\in [0:k]$,
let $r(i,A,\ell)$ be the maximum utility that the first $i$ voters can obtain from a committee $W$  
such that $W\cap v_i=A$, $|W|=\ell$, and $W\subseteq v_1\cup\ldots\cup v_{i}$.  

We have $r(1,A,|A|)=u(|A|)$ for every $A\subseteq v_1$ 
and $r(1,A,\ell)=-\infty$ for every $A\subseteq v_1$, $\ell\in[0:k]\setminus\{|A|\}$.
To compute $r(i+1,A,\ell)$ for $i\in[1:n-1]$, $A\subseteq v_{i+1}$ and $\ell\in[0:k]$,
we let $p=|A\setminus v_{i}|$ and set
$$
r(i+1,A,\ell) = \max_{D\subseteq v_i\setminus v_{i+1}} r(i, D\cup(A\cap v_i), \ell-p)+u(|A|). 
$$
Indeed, every committee $W$ with $|W|=\ell$, $W\cap v_{i+1}=A$, $W\subseteq v_1\cup\ldots\cup v_{i+1}$
contains exactly $\ell-p$ candidates from $v_1\cup\ldots\cup v_{i}$
and its intersection with $v_i$ is of the form $D\cup(A\cap v_i)$,
where candidates in $D$ are approved by $v_i$, but not $v_{i+1}$. 
We output $\max_{A\subseteq v_n}r(n,A,k)$.

This dynamic program has $n\cdot 2^s\cdot (k+1)$ states, 
and the value of each state is computed using $\calO(2^s)$ arithmetic operations.
Assuming that basic calculations take constant
time, we obtain a total runtime of $\calO(2^{2s}\cdot k\cdot n)$.
\end{proof}

A similar dynamic programming algorithm can be used if voters' preferences
satisfy CI. 

\begin{theorem}\label{thm:ci-pav-s}
Given a dichotomous profile $\calP=(v_1,\dots,v_n)$ over a candidate set
$C=\{c_1,\dots,c_m\}$ and a target committee size $k$,
if $|v_i|\le s$ for all $v_i\in \calP$ and $\calP$ satisfies CI,
then we can find a winning committee
under \pav in time $\calO(2^{s}\cdot n\cdot m)$.
\end{theorem}
\begin{proof}
Assume that $\calP$ satisfies CI with respect to the order of candidates $c_1\lhd\dots\lhd c_m$.
For each triple
$(j,A,\ell)$, where $j\in[1:m]$, $A\subseteq \{c_{j-s+1},\dots,c_j\}$, and $\ell\in [0:k]$,
let $r(j,A,\ell)$ be the maximum utility that voters can obtain from a committee $W$
such that $W\subseteq\{c_1,\dots, c_j\}$, $W\cap \{c_{j-s+1},\dots,c_j\}=A$, and $|W|=\ell$.
Also, for each $j\in[1:m-s+1]$ and each $A\subseteq \{c_j,\dots, c_{j+s-1}\}$ 
let $t(A,c_{j+s-1})=\sum_{v\in\calP: c_{j+s-1}\in v} u(|A\cap v|)$.
Note that all the quantities $t(.,.)$ can be computed in time $\calO(2^s\cdot m\cdot n)$.

We have $r(1,\emptyset,0)=0$, $r(1,\{c_1\},1)=|\{v_i: c_1\in v_i\}|$, 
and $r(1,A,\ell)=-\infty$ if $(A,\ell)\neq(\emptyset,0),(\{c_1\},1)$.
The quantities $r(j+1,A,\ell)$ for $j\in[1:m-1]$ can now be computed as follows.
If $c_{j+1}\not\in A$, we set 
$$
r(j+1,A,\ell) = \max\left\{r(j,A,\ell), r(j,A\cup\{c_{j-s}\},\ell)\right\}.
$$
Now, suppose that $c_{j+1}\in A$. Let $A'=A\setminus\{c_{j+1}\}$. 
Then $r(j+1,A,\ell)=\max\{r_1,r_2\}$ where 
\begin{align*}
r_1 &=r(j,A'\cup \{c_{j-s}\},\ell-1)-t(A',c_{j+1})+t(A,c_{j+1}),\\
r_2 &=r(j,A',\ell-1)-t(A',c_{j+1})+t(A,c_{j+1}).
\end{align*}
We output $\max_{A\subseteq\{c_{m-s+1},\dots,c_m\}}r(m,A,k)$. 
Our dynamic program has at most $2^s\cdot m\cdot (k+1)$ states, and the utility of each state
can be computed in time $\calO(1)$. Combining this with the time used to compute
$t(.,.)$, we obtain the desired bound on the running time.
\end{proof}

Our next two theorems also considers CI and VI preferences, and
deal with the case where no candidate is approved by too many voters.
Just as the algorithms in the proofs of Theorems~\ref{thm:vi-pav-s} and~\ref{thm:ci-pav-s},
the algorithms for this case are based on dynamic programming.

\begin{theorem}\label{thm:ci-pav-d}
Given a dichotomous profile $\calP=(v_1,\dots,v_n)$ over a candidate set 
$C=\{c_1,\dots,c_m\}$ and a target committee size $k$,
if $|\{i\mid c\in v_i\}|\le d$ for all $c\in C$ and $\calP$ satisfies CI,
then we can find a winning committee 
under \pav in time $\poly(d,m,n,k^d)$.
\end{theorem}

\begin{proof}
Assume without loss of generality that 
the candidate order $c_1\lhd\dots\lhd c_m$
witnesses that $\calP$ is CI.
For each voter $v_i\in \calP$, let $c_{b_i}$ and $c_{e_i}$ be, respectively, 
the first and the last candidate (with respect to $\lhd$) approved by $v_i$,
i.e., $v_i=\{c_j\mid b_i\le j\le e_i\}$.
For $j\in[1:m]$,
we say that a voter $v_i$ is {\em active} at $j$ if $b_i\le j \le e_i$;
we say that a voter $v_i$ is {\em finished} at $j$ if $e_i\le j$.
Let $B^j=\{v_i\mid b_i=j\}$, $E^j=\{v_i\mid e_i=j\}$.
Given a set $W\subseteq C$, we will refer to the quantity
$u(|W\cap v_i|)=1+1/2+\dots+1/|W\cap v_i|$ as the {\em utility of voter $i$ from set $W$}.
Throughout the proof, we make the standard assumption
that for any real-valued function $f$ we have
$\max\{f(x)\mid x\in X\}=-\infty$ when $X=\emptyset$.

Let $R(j)$ be the set of all vectors $\vecr\in[0:k]^n$ 
such that for all $\ell\in[1:n]$
it holds that $0\le r_\ell\le \min\{j-b_\ell+1, k\}$ and, 
moreover, $r_\ell=0$ whenever $v_\ell$ is not active at $c_j$.
Vectors in $R(j)$ can be used to describe the impact of a set of candidates $C$ 
in $\{c_1,\dots,c_j\}$ with $|C|\le k$ on voters who are active at $c_j$: for each $v_\ell\in \calP$,
$r_\ell$ indicates how many candidates in $C$ are approved by $v_\ell$.
As there are at most $d$ voters who are active at $j$, we have $|R(j)|\le (k+1)^d$.
For each $j\in[1:m]$, $i\in[0:\min\{j,k\}]$ and $\vecr\in R(j)$,
let $\calW(i,j,\vecr)$ be the collection of all subsets of $C$
with the following properties:
each $W\in\calW(i,j,\vecr)$ satisfies $|W|=i$, $W\subseteq \{c_1,\dots,c_j\}$,
and, moreover, for each $\ell\in[1:n]$ such that $v_\ell$ is active at $c_j$
it holds that $|v_\ell\cap W|=r_\ell$.
Intuitively, $\calW(i,j,\vecr)$ consists of all size-$i$ subsets of $\{c_1,\dots,c_j\}$
whose impact on voters who are active at $c_j$ is described by $\vecr$. 
Let $A(i,j,\vecr)$ be the maximum total utility
that voters who are finished at $j$ derive from a set in $\calW(i,j,\vecr)$;
note that $A(i,j,\vecr)=-\infty$ if $\calW(i,j,\vecr)=\emptyset$.
Clearly, it is easy to compute $A(i,1,\vecr)$ for $i\in\{0,1\}$ 
and all $\vecr\in R(1)$.

We will now explain how to compute $A(i,j,\vecr)$ given the values
of $A(i',j-1,\vecr')$ for all $i'\in[0:\min\{j-1,k\}]$ and all $\vecr\in R(j-1)$.

Suppose first that $B_j\neq\emptyset$.
By definition of $R(j)$ we have $r_x\in\{0,1\}$ for each $v_x\in B_j$.
Moreover, if we have $r_x\neq r_y$ for some $v_x,v_y\in B_j$, then 
$\calW(i,j,\vecr)=\emptyset$ and consequently $A(i,j,\vecr)=-\infty$: no subset of 
$\{c_1,\dots, c_j\}$ can intersect $v_x$, but not $v_y$ or vice versa.

Now, if $B_j\neq\emptyset$ and $r_x=1$ for all
$v_x\in B_j$, all sets in $\calW(i,j,\vecr)$ contain $c_j$,
and therefore
$$
A(i,j,\vecr)=\max_{\vecr'\in R'_1}A(i-1,j-1,\vecr')+\sum_{v_\ell\in E_j}u(r_\ell),
$$
where $R'_1$ is the set of all vectors $\vecr'\in R(j-1)$ with $r'_\ell=r'_\ell-1$
for all voters $v_\ell$ that are active at both $c_j$ and $c_{j-1}$.
Indeed, the second summand here is the total utility of voters in $E_j$;
for every such voter $v_\ell$ we know
that for any set of candidates $W\in \calW(i,j,\vecr)$ he approves exactly $r_\ell$
candidates in $W$. The first summand is the maximum total utility of voters
who are finished at $j-1$ that can be achieved by
picking a set $W'$ so that $W'\cup\{c_j\}\in \calW(i,j,\vecr)$;
every such set $W'$ is contained in $\calW(i-1,j-1,\vecr')$
for some vector $\vecr'$ in $R(j-1)$ that is consistent with $\vecr$,
i.e. satisfies $r'_\ell=r'_\ell-1$
for all voters $v_\ell$ that are active at both $c_j$ and $c_{j-1}$.

By a similar argument, if $B_j\neq\emptyset$ and $r_x=0$ for all
$v_x\in B_j$, no set in $\calW(i,j,\vecr)$ contains $c_j$,
and therefore
$$
A(i,j,\vecr)=\max_{\vecr'\in R'_0}A(i,j-1,\vecr')+\sum_{v_\ell\in E_j}u(r_\ell),
$$
where $R'_0$ is the set of all vectors $\vecr'\in R(j-1)$ with $r'_\ell=r'_\ell$
for all voters $v_\ell$ that are active at both $c_j$ and $c_{j-1}$.

Finally, suppose that $B_j=\emptyset$. Then we have to consider
both possibilities for $c_j$.
To this end, define
\begin{align*}
a_1 &=\max_{\vecr'\in R'_1}A(i-1,j-1,\vecr')+\sum_{v_\ell\in E_j}u(r_\ell),\\
a_0 &=\max_{\vecr'\in R'_0}A(i,j-1,\vecr')+\sum_{v_\ell\in E_j}u(r_\ell),
\end{align*}
where $R'_1$ and $R'_0$ are defined as above,
and set 
$$
A(i,j,\vecr)=\max\{a_1,a_0\},
$$
again, the argument for correctness is the same as above.

To complete the proof, it remains to observe that
the \pav-score of an optimal size-$k$ committee 
is given by $\max_{\vecr\in R(m)}A(k,m,\vecr)$.
Once this score is computed, the respective committee  
can be found using standard dynamic programming techniques.

To bound the running time, note that our dynamic program has $\calO(km(k+1)^d)$ variables, and
the argument above establishes that the value of $A(i,j,\vecr)$ can be computed
in time $\calO(d(k+1)^d)$ given the values of $A(i',j-1,\vecr')$
for all $i'\in[0:\min\{k, j-1\}]$, $\vecr\in R(j-1)$.
\end{proof}

\begin{theorem}\label{thm:vi-pav-d}
Given a dichotomous profile $\calP=(v_1,\dots,v_n)$ over a candidate set 
$C=\{c_1,\dots,c_m\}$ and a target committee size $k$,
if $|\{i\mid c\in v_i\}|\le d$ for all $c\in C$ and $\calP$ satisfies VI,
then we can find a winning committee 
under \pav in time $\poly(d,m,n,k^d)$.
\end{theorem}
\begin{proof}
Assume without loss of generality that
the voter order $v_1\sqsubset\dots\sqsubset v_n$
witnesses that $(C,V)$ is in VI.
For each candidate $c_j\in C$, let $v_{b_j}$ and $v_{e_j}$ be, respectively,
the first and the last voter (with respect to $\sqsubset$) who approve $c_j$,
i.e., $\{v_i\in V\mid c_j\in v_i\} =\{v_i\mid b_j\le i\le e_j\}$.
Let $C^{\ell,r}= \{c_j\mid b_j=\ell, e_j=r\}$,
$B^\ell=\{c_j\mid b_j\le \ell\}$.
Given a set $W\subseteq C$, we will refer to the quantity
$u(|W\cap v_i|)=1+1/2+\dots+1/|W\cap v_i|$ as the {\em utility of voter $i$ from set $W$}.
Throughout the proof, we make the standard assumption
that for any real-valued function $f$ we have
$\max\{f(x)\mid x\in X\}=-\infty$ when $X=\emptyset$.

Let $\calN(i)$ be the set of all $m$-by-$m$ matrices over $[0:k]$
that have the following property: for 
every matrix $N=(N_{\ell,r})_{\ell,r\in[1:m]}\in\calN(i)$,
we have $0\le N_{\ell,r}\le |C^{\ell,r}|$ if $\ell\le i\le r$
and $N_{r,\ell}=0$ if $i<\ell$ or $i>r$.
Matrices in $\calN(i)$ can be used to describe the impact of a set of candidates $W$ on voter $v_i$: 
for each $\ell, r\in[1:m]$,
$N_{\ell,r}$ indicates how many candidates in $W$ are approved by $v_i$.
Since $c_j\in v_i$ implies $i-d+1\le b_j\le i$, $i\le e_j\le i+d-1$,
we have $|\calN(i)|\le (k+1)^{d^2}$.

For each $j\in[0:k]$, $i\in[1:n]$ and $N\in \calN(i)$,
let $\calW(i,j,N)$ be the collection of all size-$j$ subsets of $B^i$
such that $v_i\cap C^{\ell,r}=N_{\ell,r}$ for all $\ell,r\in[1:m]$;
we set $\calW(i,j,N)=\emptyset$ if $j\not\in[0:k]$, $i\not\in[1:n]$ or $N\not\in \calN(i)$.
In words, $\calW(i,j,N)$ consists of all size-$j$ sets consisting of 
candidates that are approved by at least one voter in $v_1,\dots,v_i$ 
whose impact on $v_i$ is described by $N$.
Let $A(i,j,N)$ be the maximum total utility
that voters in $\{v_1,\dots,v_i\}$ derive from a set in $\calW(i,j,N)$;
note that $A(i,j,N)=-\infty$ if $\calW(i,j,N)=\emptyset$.
It is easy to compute $A(1,j,N)$ for all $j\in[0:k]$
and all $N\in \calN(1)$: we have $A(1,j,N)=u(j)$ if
$j=\sum_{r\in[1:m]}n_{1, r}$ and $A(1,j,N)=-\infty$ otherwise.
Also, for each $i\in[1:n]$ we have $A(i,0,N)=0$ if $N_{\ell,r}=0$ for all $\ell,r\in[1:m]$
and $A(i,0,N)=-\infty$ otherwise.

We will now explain how to compute $A(i,j,N)$ given the values
of $A(i-1,j',N)$ for all $j'\in[1,j]$ and all $N\in\calN(i)$.
Fix $i\in[2:n]$, $j\in[0:k]$, $N\in\calN(i)$.
Note first that for any set $W\in \calW(i,j,N)$ we have
$$
|v_i\cap W| = \sum_{r,\ell\in[1:m]}n_{r,\ell};
$$
also, if $\sum_{r,\ell\in[1:m]}n_{r,\ell}> j$, then $\calW(i,j, N)=\emptyset$. 

Further, for every set $W\in\calW(i,j,N)$ the set $W\setminus\{c_t\mid b_t=i\}$
belongs to $\calW(i-1,j', N')$ for $j'=j-|\{c_t\mid b_t=i\}|$ and
for some matrix $N'\in\calN'(i)$ with $n'_{\ell,r}=n_{\ell,r}$ for $\ell\neq i$ and $r\neq i-1$.
Let $j'=j-|\{c_t\mid b_t=i\}|$,
$\calN'=\{N'\in\calN(i-1)\mid n'_{\ell,r}=n_{\ell,r}\text{ for }\ell\neq i,r\neq i-1\}$.
Then we have
$$
A(i,j,N) = \max_{N'\in\calN'(i,N)} A(i-1,j', N') + u(\sum_{r,\ell\in[1:m]}n_{r,\ell})
$$ 
if $\sum_{r,\ell\in[1:m]}n_{r,\ell}\le j$ and $A(i,j, N)=-\infty$ otherwise.

To complete the proof, it remains to observe that
the \pav-score of an optimal size-$k$ committee
is given by $\max_{N\in \calN(n)}A(n,k,N)$.
Once this score is computed, the respective committee
can be found using standard dynamic programming techniques,
and the bound on running time follows immediately.
\end{proof}

The reader may wonder if constraints on $s$ and $d$ in Theorems~\ref{thm:vi-pav-s},~\ref{thm:ci-pav-s}
and~\ref{thm:ci-pav-d} are necessary. We conjecture that the answer is yes, i.e., winner determination under
\pav remains hard under CI and VI preferences.

\begin{conjecture}
\pav is NP-hard even for CI and VI preferences.
\end{conjecture}

However, for ``truncated'' weight vectors $\vecw$ we can find $\vecw$-\pav
winners in polynomial time. As the 
$(1,0,\dots)$-\pav rule is essentially the classic Chamberlin--Courant rule \citep{ChCo83a}
for dichotomous preferences, our next result can be seen as an extension of the results of \citep{BSU13a}
and \citep{sko-yu-fal:j:mwsc} for the Chamberlin--Courant rule and single-peaked
and single-crossing preferences: while we work on a less expressive domain
(dichotomous preferences vs. total orders), we can handle a larger class of rules
(all weight vectors with a constant number of non-zero entries rather than just $(1,0,\dots,)$).

\begin{theorem}
Consider a weight vector $\vecw$ where $w_i=0$ for $i>i_0$ for some constant $i_0$.
Then given a dichotomous profile $\calP=(v_1,\dots,v_n)$ over a candidate set
$C=\{c_1,\dots,c_m\}$ and a target committee size $k$, if $\calP$
satisfies VI, we can find a winning committee
under $\vecw$-\pav in polynomial time.
\end{theorem}

\begin{proof}
Assume that $\calP$ satisfies VI with respect to the order of voters $v_1\sqsubset\dots\sqsubset v_n$. 

The following algorithm is a refinement of Theorem~\ref{thm:vi-pav-s}.
For each triple
$(i,A,\ell)$, where $i\in[1:n]$, $A\subseteq v_i$, and $\ell\in [0:k]$,
let $r(i,A,\ell)$ be the maximum utility that the first $i$ voters can obtain from a committee $W$  
such that $|W|=\ell$, and $W\subseteq v_1\cup\ldots\cup v_{i}$ and $A\subseteq W$.  

We have $r(1,A,\ell)=u(\ell)$ for every $\ell\in [0:|v_1|]$ and $A\subseteq v_1$ with $|A|= \min(i_0,\ell)$. 
In addition, we have $r(1,A,\ell)=-\infty$ for every other $A\subseteq v_1$ and $\ell\in[0:k]$.
To compute $r(i+1,A,\ell)$ for $i\in[1:n-1]$, $A\subseteq v_{i+1}$ with $|A|\leq i_0$ and $\ell\in[|A|:k]$,
we 
let $s=|v_{i+1}\setminus (v_{i}\cup A)|$, i.e., the maximal number of candidates that might have been added in the $i+1$st step to the committee but that do not show up in $A$, and
set
$$
r(i+1,A,\ell) = \max_{} r(i, D\cup(A\cap v_i), \ell-|A|-r)+u(|A|), 
$$
where the maximum is taken over all $D\subseteq v_i\setminus v_{i+1}$ with $|D|\in [0: i_0-|A\cap v_i|]$ and all $r\in[0:s]$.

This dynamic program has $n\cdot m^{i_0}\cdot (k+1)$ states, 
and the value of each state is computed using $\calO(m^{i_0}+1)$ arithmetic operations.
Assuming that basic calculations take constant
time, we obtain a total runtime of $\calO(n\cdot m^{2i_0+1}\cdot k)$, which is polynomial for constant $i_0$.
\end{proof}

\begin{theorem}
Consider a weight vector $\vecw$ where $w_i=0$ for $i>i_0$ for some constant $i_0$.
Then given a dichotomous profile $\calP=(v_1,\dots,v_n)$ over a candidate set
$C=\{c_1,\dots,c_m\}$ and a target committee size $k$, if $\calP$
satisfies CI, we can find a winning committee
under $\vecw$-\pav in polynomial time.
\end{theorem}

\begin{proof}
Assume that $\calP$ satisfies CI with respect to the order of candidates $c_1\lhd\dots\lhd c_m$.
The following algorithm is a refinement of Theorem~\ref{thm:ci-pav-s}.
For two sets $C_1,C_2\subseteq C$ we write $C_1\lhd C_2$ to denote that for all $c\in C_1$ and $d\in C_2$ it holds that $c\lhd d$.
For each triple $(j,A,\ell)$, where $j\in[1:m]$, $A\subseteq \{c_1,\dots,c_j\}$, $|A|\leq i_0$ and $\ell\in [0:k]$,
let $r(j,A,\ell)$ be the maximum utility that voters can obtain from a committee $W$
such that $A\subseteq W$, $|W|=\ell$ and $W\setminus A\lhd  A$.
Also, for each $j\in[1:m-s+1]$ and each $A\subseteq \{c_j,\dots, c_{j+s-1}\}$ 
let $t(A,c)=\sum_{v\in\calP: c\in v} u(|A\cap v|)$.
It is essential that, given a committee $W$ satisfying the conditions above, $t(W,c)=t(A,c)$ assuming CI preferences and $c'\lhd c$ for all $c'\in A\setminus\{c\}$.
Furthermore, note that all the quantities $t(.,.)$ can be computed in time $\calO(n\cdot m^{i_0+1})$ since we assume that $u(.)$ can be computed in constant time.

We have $r(1,\emptyset,0)=0$, $r(1,\{c_1\},1)=|\{v_i: c_1\in v_i\}|$,
and $r(1,A,\ell)=-\infty$ if $(A,\ell)\neq(\emptyset,0),(\{c_1\},1)$.
The quantities $r(j+1,A,\ell)$ for $j\in[1:m-1]$ and $A\subseteq \{c_1,\dots,c_{j+1}\}$ with $|A|\leq i_0$ can now be computed as follows.
If $c_{j+1}\not\in A$, we set 
$$
r(j+1,A,\ell) = r(j,A,\ell).
$$
Now, suppose that $c_{j+1}\in A$. Let $A'=A\setminus\{c_{j+1}\}$. 
Then $r(j+1,A,\ell)=\max\{r_1,r_2\}-t(A',c_{j+1})+t(A,c_{j+1})$ where 
\begin{align*}
r_1 &=\max_{c \in C\text{ with }\{c\}\lhd A} r(j,A'\cup \{c\},\ell-1),\\
r_2 &=r(j,A',\ell-1).
\end{align*}
We output $\max_{A\subseteq C\text{ with }|A|\leq i_0}r(m,A,k)$. 
Our dynamic program has at most $m^{i_0+1}\cdot (k+1)$ states, and the utility of each state
can be computed in time $\calO(m)$. Combining this with the time used to compute
$t(.,.)$, we obtain a total runtime of $\calO(n\cdot m^{i_0+1}\cdot k)$, which is polynomial for fixed $i_0$.
\end{proof}

Moreover, for the more restricted domains, such as VEI, CEI, WSC and PART   
we can design polynomial-time algorithms for both \mav and \pav,
under no additional constraints on preferences (again, 
our results extend to $\vecw$-\pav).

\begin{theorem}\label{thm:pav-mav-vei}
Given a dichotomous profile $\calP=(v_1,\dots,v_n)$ over a candidate set
$C=\{c_1,\dots,c_m\}$ and a target committee size $k$, if $\calP$
satisfies VEI, we can find a winning committee 
under \mav and \pav in polynomial time.
\end{theorem}
\begin{proof}
Assume without loss of generality that $\calP$ satisfies
VEI for voter order $v_1\sqsubset\dots\sqsubset v_n$.
Each candidate in $C$ belongs to one 
of the following four groups:
$C_1=v_1\cap v_n$, $C_2=v_1\setminus v_n$, $C_3=v_n\setminus v_1$, 
and $C_4 = \overline{v_1}\cap \overline{v_n}$; candidates in $C_1$
are approved by all voters and candidates in $C_4$ are not approved
by any of the voters. 

Suppose first that $|C_1\cup C_2\cup C_3|< k$.
Then there exists an optimal committee for both \pav and \mav
that contains all candidates in $C_1\cup C_2\cup C_3$
and exactly $k-|C_1\cup C_2\cup C_3|$ candidates from $C_4$.
Hence, we can now assume that this is not the case. 
Then there exist an optimal committee
that contains no candidates from $C_4$.

Now, if $|C_1|\ge k$, an optimal committee for both \pav and \mav consists
of $k$ candidates from $C_1$, and if $|C_1|< k$, there exists
an optimal committee that contains all candidates in $C_1$.
It remains to decide how to allocate the remaining places 
among candidates in $C_2$ and $C_3$. To do so, we observe that
there is a natural ordering over each of these sets: given a pair of candidates
$(c, c')$ in $C_2\times C_2$ or $C_3\times C_3$, 
we write $c\le c'$ if $\{i: c\in v_i\}\subseteq \{i: c'\in v_i\}$.
Note that every two candidates in $C_2$ are comparable with respect to 
$\le$, and so are every two candidates in $C_3$.
It is now easy to see that there exists an optimal committee (for \pav or \mav)
that consists of candidates in $C_1$,
top $p$ candidates in $C_2$ with respect to $\le$ 
and top $r$ candidates in $C_3$ with respect to $\le$
for some non-negative values of $p,r$ with $p+r+|C_1|=k$.
Thus, by considering at most $k^2$ possibilities for $p$ and $r$,
we can find an optimal committee.
\end{proof}

For CEI, we employ a dynamic programing algorithm, somewhat similar to the one used in Theorem~\ref{thm:ci-pav-s}.
Since we consider a more constrained preferences (CEI instead of CI), we do not require to maintain an exponential number of states.

\begin{theorem}\label{thm:pav-mav-cei}
Given a dichotomous profile $\calP=(v_1,\dots,v_n)$ over a candidate set
$C=\{c_1,\dots,c_m\}$ and a target committee size $k$, if $\calP$
satisfies CEI, we can find a winning committee 
under \mav and \pav in polynomial time.
\end{theorem}
\begin{proof}
Assume that $\calP$ satisfies CEI with respect to the order of candidates $c_1\lhd\dots\lhd c_m$.
Let us consider PAV.
For $j\in[1:m]$, let $V_j$ denote all votes $v\in \calP$ such that $c_j$ is the rightmost approved candidate of $v$ if $c_1\in v$ and $c_j$ is the leftmost approved candidate of $v$ if $c_1\notin v$.
States are identified by a pair $(j,\ell)$, where $j\in[1:m]$ and $\ell\in[0:k]$.
Let $r(j,\ell)$ be the maximum utility that the voters $V_1\cup\cdots\cup V_j$ can obtain from a committee $W$ with $|W|=k$ such that $|W\cap \{c_1,\ldots,c_j\}|=\ell$, i.e., $W$ contains $\ell$ candidates to the left of $c_j$ (including $c_j$) and $k-\ell$ candidates to the strictly to the right of $c_j$.

We have $r(1,0)=0$ and $r(1,1)=|V_1|$, since $V_1$ contains exactly those votes that approve only $c_1$.
For $j\in[1:m]$ and $\ell\in[0:k]$ we have $r(j,\ell)=-\infty$ if $\ell>j$ or if $k-\ell>m-j$.
The remaining quantities $r(j+1,\ell)$ for $j\in[1:m-1]$ can be computed as follows:
Let $V_{j+1}^-=\{v\in V_{j+1}: c_{j}\in v\}$ and $V_{j+1}^+=\{v\in V_{j+1}: c_{j}\notin v\}$.
The quantity $r(j+1,\ell)=max(r_1,r_2)$, where
\begin{align*}
r_1 &= r(j,\ell)+\sum_{v\in V_{j+1}^-} u(\ell)+\sum_{v\in V_{j+1}^+} u(k-\ell)\quad\text{ and} \\
r_2 &= r(j,\ell-1)+\sum_{v\in V_{j+1}^-} u(\ell)+\sum_{v\in V_{j+1}^+} u(k-\ell+1).
\end{align*}
Here, $r_1$ corresponds to committees that do not contain $c_{j+1}$ and $r_2$ to committees that contain $c_{j+1}$.
We output $r(m,k)$.
These quantities can be computed in polynomial time. 

For MAV we use a similar approach.
For $j\in[1:m]$ and $\ell\in[0:k]$, let $g(j,\ell)$ be the minimum MAV-score obtainable by the voters in $V_1\cup\cdots\cup V_j$ from a committee $W$ with $|W|=k$ such that $|W\cap \{c_1,\ldots,c_j\}|=\ell$.
Recall that $V_1$ contains only votes that approve $c_1$; hence we have $g(1,0)=0$, $g(1,1)=1$ if $|V_1|\geq1$ and $g(1,0)=1$ and $g(1,1)=0$ if $V_1=\emptyset$.
For $j\in[1:m]$ and $\ell\in[0:k]$ we have $g(j,\ell)=\infty$ if $\ell>j$ or if $k-\ell>m-j$.
The remaining quantities $g(j+1,\ell)$ for $j\in[1:m-1]$ can be computed as follows:
Observe that for $v\in V_{j+1}^-$ and a committee $W$ with $|W|=k$ such that $|W\cap \{c_1,\ldots,c_j\}|=\ell$, $|W\setminus v|+|v\setminus W|=(j+1-\ell)+(k-\ell)=j+1+k-2\ell$.
For $v\in V_{j+1}^+$, it holds that $|W\setminus v|+|v\setminus W|=(m-j-1-k+\ell)+(\ell)=m-(j+1+k-2\ell)$ if $c_{j+1}\notin v$ and $|W\setminus v|+|v\setminus W|=(m-j-1-k+\ell-1)+(\ell-1)=m-(j+3+k-2\ell)$ if $c_{j+1}\in v$.
The quantity $g(j+1,\ell)=\min\{g_1,g_2\}$, where
{\setlength{\medmuskip}{2mu}
\begin{align*}
g_1 &= \max\{r(j,\ell),j+1+k-2\ell,m-(j+1+k-2\ell)\}\text{ and}\\
g_2 &= \max\{r(j,\ell-1),j+1+k-2\ell,m-(j+3+k-2\ell)\}.
\end{align*}
}
As before, $g_1$ corresponds to committees that do not contain $c_{j+1}$ and $g_2$ to committees that contain $c_{j+1}$.
We output $g(m,k)$.
\end{proof}

\begin{proposition}\label{prop:pav-mav-wsc-part}
Given a dichotomous profile $\calP=(v_1,\dots,v_n)$ over a candidate set
$C=\{c_1,\dots,c_m\}$ and a target committee size $k$, if $\calP$
satisfies WSC or PART, we can find a winning committee 
under \mav and \pav in polynomial time.
\end{proposition}
\begin{proof}[Proof sketch]
For WSC, we can use the characterization in Lemma~\ref{lem:WSC-char};
the problem then boils down to deciding how many candidates 
to select from each of the sets $u\setminus w$, $u\cap w$ and $w\setminus u$. 
For PART and \pav,
we can show that an optimal committee can be found by a  
natural greedy algorithm that at each point selects
the candidate with the largest ``marginal contribution'' to the total utility.
For PART and \mav, we check, for each $t=0,\dots,n$,
whether there exists a committee whose \mav-score is at most $t$.
This is the case if for each voter $v\in\calP$ we can select at least 
$(|v|+k-t)/2$ candidates from $v$. Thus, if $v_1,\dots,v_\ell$ are the distinct
votes in $\calP$, we need to check that $\sum_{i=1}^\ell |v_i|\le \ell t-(\ell-2)k$. 
\end{proof}


\section{Conclusions and Open Problems}
We have initiated research on analogues of 
the notions of single-peakedness and single-crossingness
for dichotomous preference domains. We have proposed many constraints
that capture some aspects of what it means for dichotomous preferences to be single-dimensional,
explored the relationship among them, and showed that these constraints 
can be useful for identifying efficiently solvable special cases
of hard voting problems on dichotomous domains. The algorithmic 
results in Section~\ref{sec:algorithms} can be seen as a proof that our 
approach has merit; however, there is certainly room for improvement there,
both in terms of removing restrictions on the sizes of approval sets 
and number of voters that approve each candidate (for \pav) 
and in terms of considering larger domains, such as PSC for \pav
and CI/VI for \mav.

For many of our constraints, we have provided efficient algorithms for checking whether
a given dichotomous profile satisfies that constraint; only checking of PSC/SSC remains as an open case.
We can also ask if it is possible
to detect if a given dichotomous profile is close to satisfying a structural
constraint, and whether such ``almost-structured'' profiles have useful algorithmic 
properties; similar issues for profiles of total orders 
have recently received a lot of attention in the literature
\citep{cor-gal-spa:c:sp-width,cor-gal-spa:c:spsc-width,bre-che-woe:c:distance,erd-lac-pfa:c:nearly-sp,elk-lac:c:nearly,FHH14}.

\section*{Acknowledgments}

This work was supported by the European Research Council (ERC) under grant number 639945
(ACCORD).
The second author was additionally supported by the Austrian Science Foundation FWF, grant P25518-N23 and Y698.

\bibliographystyle{plainnat}
\bibliography{approval}

\begin{thebibliography}{30}
\providecommand{\natexlab}[1]{#1}
\providecommand{\url}[1]{\texttt{#1}}
\expandafter\ifx\csname urlstyle\endcsname\relax
  \providecommand{\doi}[1]{doi: #1}\else
  \providecommand{\doi}{doi: \begingroup \urlstyle{rm}\Url}\fi

\bibitem[Aziz et~al.(2015{\natexlab{a}})Aziz, Gaspers, Gudmundsson, Mackenzie,
  Mattei, and Walsh]{AGG+14a}
H.~Aziz, S.~Gaspers, J.~Gudmundsson, S.~Mackenzie, N.~Mattei, and T.~Walsh.
\newblock Computational aspects of multi-winner approval voting.
\newblock In \emph{AAMAS'15}, pages 107--115, 2015{\natexlab{a}}.

\bibitem[Aziz et~al.(2015{\natexlab{b}})Aziz, Brill, Conitzer, Elkind, Freeman,
  and Walsh]{aaai/AzizBCEFW15}
Haris Aziz, Markus Brill, Vincent Conitzer, Edith Elkind, Rupert Freeman, and
  Toby Walsh.
\newblock Justified representation in approval-based committee voting.
\newblock In \emph{AAAI'15}, pages 784--790, 2015{\natexlab{b}}.

\bibitem[Betzler et~al.(2013)Betzler, Slinko, and Uhlmann]{BSU13a}
N.~Betzler, A.~Slinko, and J.~Uhlmann.
\newblock On the computation of fully proportional representation.
\newblock \emph{Journal of Artificial Intelligence Research}, 47:\penalty0
  475--519, 2013.

\bibitem[Black(1958)]{bla:b:sp}
D.~Black.
\newblock \emph{The Theory of Committees and Elections}.
\newblock Cambridge University Press, 1958.

\bibitem[Booth and Lueker(1976)]{boo-lue:j:cons1s}
K.~Booth and G.~Lueker.
\newblock Testing for the consecutive ones property, interval graphs, and graph
  planarity using {PQ}-tree algorithms.
\newblock \emph{Journal of Computer and System Sciences}, 13\penalty0
  (3):\penalty0 335--379, 1976.

\bibitem[Brams et~al.(2007)Brams, Kilgour, and Sanver]{BKS07a}
S.~Brams, D.~M. Kilgour, and R.~M. Sanver.
\newblock A minimax procedure for electing committees.
\newblock \emph{Public Choice}, 132\penalty0 (3-4):\penalty0 401--420, 2007.

\bibitem[Brandt et~al.(2010)Brandt, Brill, Hemaspaandra, and
  Hemaspaandra]{BBHH10}
F.~Brandt, M.~Brill, E.~Hemaspaandra, and L.~Hemaspaandra.
\newblock Bypassing combinatorial protections: Polynomial-time algorithms for
  single-peaked electorates.
\newblock In \emph{AAAI'10}, pages 715--722, 2010.

\bibitem[Brandt et~al.(2015)Brandt, Conitzer, Endriss, Lang, and
  Procaccia]{comsoc-book}
F.~Brandt, V.~Conitzer, U.~Endriss, J.~Lang, and A.~D. Procaccia, editors.
\newblock \emph{Handbook of Computational Social Choice}.
\newblock Cambridge University Press, 2015.

\bibitem[Bredereck et~al.(2013)Bredereck, Chen, and
  Woeginger]{bre-che-woe:c:distance}
R.~Bredereck, J.~Chen, and G.~Woeginger.
\newblock Are there any nicely structured preference profiles nearby?
\newblock In \emph{IJCAI'13}, pages 62--68, 2013.

\bibitem[Chamberlin and Courant(1983)]{ChCo83a}
B.~Chamberlin and P.~Courant.
\newblock Representative deliberations and representative decisions:
  Proportional representation and the {B}orda rule.
\newblock \emph{American Political Science Review}, 77\penalty0 (3):\penalty0
  718--733, 1983.

\bibitem[Cornaz et~al.(2012)Cornaz, Galand, and
  Spanjaard]{cor-gal-spa:c:sp-width}
D.~Cornaz, L.~Galand, and O.~Spanjaard.
\newblock Bounded single-peaked width and proportional representation.
\newblock In \emph{ECAI'12}, pages 270--275, 2012.

\bibitem[Cornaz et~al.(2013)Cornaz, Galand, and
  Spanjaard]{cor-gal-spa:c:spsc-width}
D.~Cornaz, L.~Galand, and O.~Spanjaard.
\newblock Kemeny elections with bounded single-peaked or single-crossing width.
\newblock In \emph{IJCAI'13}, pages 76--82, 2013.

\bibitem[Doignon and Falmagne(1994)]{DoignonF94}
J.{-}P. Doignon and J.{-}C. Falmagne.
\newblock A polynomial time algorithm for unidimensional unfolding
  representations.
\newblock \emph{Journal of Algorithms}, 16\penalty0 (2):\penalty0 218--233,
  1994.

\bibitem[Elkind and Lackner(2014)]{elk-lac:c:nearly}
E.~Elkind and M.~Lackner.
\newblock On detecting nearly structured preference profiles.
\newblock In \emph{AAAI'14}, pages 661--667, 2014.

\bibitem[Elkind et~al.(2014)Elkind, Faliszewski, and
  Skowron]{elk-fal-sko:c:spsc}
E.~Elkind, P.~Faliszewski, and P.~Skowron.
\newblock A characterization of the single-peaked single-crossing domain.
\newblock In \emph{AAAI'14}, pages 654--660, 2014.

\bibitem[Elkind et~al.(2015)Elkind, Faliszewski, Lackner, and
  Obraztsova]{EFLO15}
E.~Elkind, P.~Faliszewski, M.~Lackner, and S.~Obraztsova.
\newblock The complexity of recognizing incomplete single-crossing preferences.
\newblock In \emph{AAAI'15}, pages 865--871, 2015.

\bibitem[Elkind and Lackner(2015)]{aaai/ElkindL15-dichpref}
Edith Elkind and Martin Lackner.
\newblock Structure in dichotomous preferences.
\newblock In \emph{IJCAI'15}, pages 2019--2025, 2015.

\bibitem[Erd\'{e}lyi et~al.(2013)Erd\'{e}lyi, Lackner, and
  Pfandler]{erd-lac-pfa:c:nearly-sp}
G.~Erd\'{e}lyi, M.~Lackner, and A.~Pfandler.
\newblock Computational aspects of nearly single-peaked electorates.
\newblock In \emph{IJCAI'13}, pages 283--289, 2013.

\bibitem[Faliszewski et~al.(2011)Faliszewski, Hemaspaandra, Hemaspaandra, and
  Rothe]{fal-hem-hem-rot:j:sp}
P.~Faliszewski, E.~Hemaspaandra, L.~Hemaspaandra, and J.~Rothe.
\newblock The shield that never was: {Societies} with single-peaked preferences
  are more open to manipulation and control.
\newblock \emph{Information and Computation}, 209\penalty0 (2):\penalty0
  89--107, 2011.

\bibitem[Faliszewski et~al.(2014)Faliszewski, Hemaspaandra, and
  Hemaspaandra]{FHH14}
P.~Faliszewski, E.~Hemaspaandra, and L.~Hemaspaandra.
\newblock The complexity of manipulative attacks in nearly single-peaked
  electorates.
\newblock \emph{Artificial Intelligence}, 207:\penalty0 69--99, 2014.

\bibitem[Kilgour(2010)]{Kilg10a}
D.~M. Kilgour.
\newblock Approval balloting for multi-winner elections.
\newblock In \emph{Handbook on Approval Voting}, chapter~6. Springer, 2010.

\bibitem[Kilgour and Marshall(2012)]{MaKi12a}
D.~M. Kilgour and E.~Marshall.
\newblock Approval balloting for fixed-size committees.
\newblock In D.~S. Felsenthal and M.~Machover, editors, \emph{Electoral
  Systems}, Studies in Choice and Welfare, pages 305--326. Springer, 2012.

\bibitem[Lackner(2014)]{La14}
M.~Lackner.
\newblock Incomplete preferences in single-peaked electorates.
\newblock In \emph{AAAI'14}, pages 742--748, 2014.

\bibitem[LeGrand et~al.(2007)LeGrand, Markakis, and Mehta]{LMM07a}
R.~LeGrand, E.~Markakis, and A.~Mehta.
\newblock Some results on approximating the minimax solution in approval
  voting.
\newblock In \emph{AAMAS'07}, pages 1185--1187, 2007.

\bibitem[Magiera and Faliszewski(2014)]{MF14}
K.~Magiera and P.~Faliszewski.
\newblock How hard is control in single-crossing elections?
\newblock In \emph{ECAI'14}, pages 579--584, 2014.

\bibitem[Mirrlees(1971)]{mir:j:single-crossing}
J.~Mirrlees.
\newblock An exploration in the theory of optimal income taxation.
\newblock \emph{Review of Economic Studies}, 38:\penalty0 175--208, 1971.

\bibitem[Nederlof and Woeginger(2015)]{woeginger-nederlof:personal-commun}
J.~Nederlof and G.~Woeginger.
\newblock Personal communication.
\newblock May 2015.

\bibitem[Skowron et~al.(2015{\natexlab{a}})Skowron, Faliszewski, and
  Lang]{SFL15}
P.~Skowron, P.~Faliszewski, and J.~Lang.
\newblock Finding a collective set of items: From proportional
  multirepresentation to group recommendation.
\newblock In \emph{AAAI'15}, 2015{\natexlab{a}}.

\bibitem[Skowron et~al.(2015{\natexlab{b}})Skowron, Yu, Faliszewski, and
  Elkind]{sko-yu-fal:j:mwsc}
P.~Skowron, L.~Yu, P.~Faliszewski, and E.~Elkind.
\newblock The complexity of fully proportional representation for
  single-crossing electorates.
\newblock \emph{Theoretical Computer Science}, 569:\penalty0 43--57,
  2015{\natexlab{b}}.

\bibitem[Walsh(2007)]{wal:c:uncertainty}
T.~Walsh.
\newblock Uncertainty in preference elicitation and aggregation.
\newblock In \emph{AAAI'07}, pages 3--8, 2007.

\end{thebibliography}

\end{document}